\documentclass[a4paper, 11pt]{amsart}
\usepackage{amsmath, amsthm, amscd, amssymb, amsfonts, amsxtra, amssymb, latexsym, bm}
\usepackage[shortlabels]{enumitem} 
\usepackage{multirow}
\usepackage{verbatim}
\usepackage{graphicx,epsfig,tikz}
\usepackage{float}
\usepackage{hyperref}
\hypersetup{colorlinks = true,	allcolors  = blue}

\newcommand{\N}{\mathbb{N}}
\newcommand{\ff}{\mathbb{F}}
\newcommand{\Tr}{\operatorname{Tr}}

\newcommand{\CC}{\mathcal{C}}

\newcommand{\G}{\Gamma}
\newcommand{\Div}{{\rm Div}}

\DeclareMathOperator{\lcm}{lcm} 

\newcommand{\sk}{\smallskip}
\newcommand{\msk}{\medskip}

\newtheorem{thm}{Theorem}[section]
\newtheorem{prop}[thm]{Proposition}
\newtheorem{lem}[thm]{Lemma}
\newtheorem{coro}[thm]{Corollary}

\theoremstyle{definition}
\newtheorem{rem}[thm]{Remark}
\newtheorem{exam}[thm]{Example}
\newtheorem{defi}[thm]{Definition}

\theoremstyle{remark}

\counterwithout{table}{section}

\usepackage{color}

\begin{document} \sloppy
	\numberwithin{equation}{section}
	\title[The weight enumerators of ICCs of length $2^\ell$]{The weight enumerators of irreducible \\ cyclic codes of length $2^\ell$}
	\author{Ricardo A.\@ Podest\'a, Denis E.\@ Videla}
	\dedicatory{\today}
	\thanks{Partially supported by CONICET, FonCyT and SECyT-UNC}
	\address{Ricardo A.\@ Podest\'a. FaMAF -- CIEM (CONICET), Universidad Nacional de C\'ordoba. 
		\newline Av.\@ Medina Allende 2144, Ciudad Universitaria, (5000), C\'ordoba, Argentina. \newline
		{\it E-mail: podesta@famaf.unc.edu.ar}}
	\address{Denis E.\@ Videla. FaMAF -- CIEM (CONICET), Universidad Nacional de C\'ordoba. 
		\newline	Av.\@ Medina Allende 2144, Ciudad Universitaria, (5000), C\'ordoba, Argentina.
		\newline {\it E-mail: devidela@famaf.unc.edu.ar}}

	\begin{abstract} 
		Using the spectral relation between irreducible cyclic codes (ICC) and generalized Paley (GP) graphs recently found by the authors, we study the weight distribution of irreducible cyclic codes of length $2^\ell$, $\ell \in \N$. These codes were previously studied by Sharma, Bakshi and Raka, who divided them into two families (each having two cases).
		First we characterize the GP-graphs spectrally related to these families of ICCs: they are complete, semiprimitive, Hamming, or Cartesian products of semiprimitive GP-graphs. 
		We then compute the spectrum of these graphs and hence, by the mentioned spectral relation we finally obtain neat explicit expressions for the weight enumerators of all the irreducible cyclic codes of length $2^\ell$ with  $\ell \in \N$.
	\end{abstract}

	\maketitle

	\section{Introduction: Preliminaries and results} \label{sec:1}
	In this work we will obtain closed explicit expressions for the weight enumerators of all $r$-ary irreducible cyclic codes of length $2^\ell$ for $\ell \in \N$.
	
	\subsection*{Motivation}
	In 2009, Sharma, Bakshi and Raka (\cite{SBR}) studied the family of $r$-ary irreducible cyclic codes of length $2^\ell$. 
	They divided these codes into four cases and studied their weight distribution. For half of the cases they give explicit expressions for the weight enumerators, while for the other half they give a somewhat complicated summation formula for the enumerators. 	
	
	On the other hand, this year we gave a spectral relation between $r$-ary irreducible cyclic codes and generalized Paley graphs (\cite{PV26b}). 
	We will obtain explicit weight enumerators for all the $r$-ary irreducible cyclic codes of length $2^\ell$ through this spectral relation. We will first characterize those GP-graphs associated to the four families given by Sharma et al., we then compute the spectrum of these graphs, and finally, using the relation, obtain the expression for the weight enumerators. 
	In this way, not only do we obtain the weight enumerators of the codes using another technique, but also we will obtain explicit expressions for the weight enumerators in all the cases, improving in this way the original result of Sharma et al.

	\subsection*{Irreducible cyclic codes and GP-graphs} 
	Now we describe the objects of study in this work (irreducible cyclic codes and generalized Paley graphs), fix some notations for the rest, and recall the main result relating the weight enumerator of the irreducible cyclic codes and the spectrum of GP-graphs.

	\subsubsection*{GP-graphs} 
	If $G$ is an abelian group and $S$ is a subset of $G$ (the connection set) not containing $0$, the associated Cayley graph $\Gamma = Cay(G,S)$ is the directed graph (digraph) with vertex set $G$ and where two vertices $u,v$ form a directed edge from $u$ to $v$ in $\Gamma$ if and only if $v-u \in S$. If $S$ is symmetric ($S=-S$), then $Cay(G,S)$ is a simple (undirected) graph. 
	Taking a finite field as the group $G$ and $k$-th powers in the field as connection set we get the GP-graphs.
	\begin{defi} \label{def: GP-graphs}
		Let $\ff_{q}$ denote any finite field of $q$ elements, say $q=p^m$ with $p$ prime and $m \in \N$. 
		The \textit{generalized Paley graph} is the Cayley graph (\textit{GP-graph} for short) 
		\begin{equation} \label{eq: Gkq}
			\G(k,q) = Cay(\ff_{q},R_{k}) \qquad \text{ with } \qquad R_{k} = \{ x^{k} : x \in \ff_{q}^*\}.
	\end{equation} 
	That is, $\G(k,q)$ is the graph with vertex set $\ff_{q}$ and two vertices $u,v \in \ff_{q}$ are neighbors (directed edge) 
	if and only if 
		$v-u=x^k$ 
	for some $x\in \ff_{q}^*$. 
\end{defi}

Let $\G=(V,E)$ be a graph with $n$ vertices, where $V$ denotes the vertex set and $E$ the edge set. The eigenvalues of $\Gamma$ are the eigenvalues $\{\lambda_i\}_{i=1}^{n}$ of its adjacency matrix $A$. The \textit{spectrum} of $\Gamma$, denoted
\begin{equation} \label{eq: spec(G)}
	Spec(\Gamma)=\{[\lambda_{i_1}]^{m_{i_1}},\cdots,[\lambda_{i_s}]^{m_{i_s}}\}
\end{equation}	
is the (multi)set of all the different eigenvalues $\{\lambda_{i_j}\}$ of $\Gamma$, counted with their multiplicities $\{m_{i_j}\}$, that is $m(\lambda_{i_j})=m_{i_j}$. 
The spectrum of $\G$ is said to be \textit{integral} if $Spec(\Gamma) \subset\mathbb{Z}$. 
We recall that if $\Gamma$ is a $k$-regular graph, then $\lambda_1 = k$. Moreover, the multiplicity of $\lambda_1$ equals the number of connected components of $G$ and hence $\Gamma$ is connected if and only if $m(\lambda_1) = 1$. 

GP-graphs have been extensively studied in the past few years. 
For instance, Lim and Praeger studied their automorphism groups and characterized all GP-graphs which are Hamming graphs (\cite{LP}, 2009). 
Also, Pearce and Praeger characterized all GP-graphs which are Cartesian decomposable (\cite{PP}, 2019); some extensions of this decomposition to the directed case can be found in \cite{PV22}. The clique number of $\G(k,q)$ was studied by Cohen and Yip (\cite{Co}, \cite{Y1}).  
Finally, the diameter of these graphs and its relation with the Waring problem over finite fields was recently approached by us in \cite{PV21} and \cite{PV22}. 
The spectrum of generalized Paley graphs was studied in \cite{PV25b} (and others), where we showed that the eigenvalues are exactly the Gaussian periods. 
In that work, we also characterized the GP-graphs with integral spectrum: they are those GP-graphs $\G(k,q)$ 
that satisfy the arithmetic condition $k\mid \frac{q-1}{p-1}$. 
In the works \cite{PV23} and \cite{PV26b}, we study the spectral relation between GP-graphs and irreducible cyclic codes.

\subsubsection*{Irreducible cyclic codes}
A \textit{linear code} of length $n$ over $\ff_r$ is a vector subspace $\CC \subseteq \ff_r^n$. If $\CC$ has dimension $h$ and minimum Hamming distance $d$, we say that it is an $[n,h,d]_r$-code. 
The weight of a codeword $c=(c_{0},\ldots,c_{n-1})$ is the number $w(c)$ of its nonzero coordinates. 
The \textit{weight distribution} of $\CC$, sometimes also called \textit{spectrum},
is the sequence $Spec(\CC) = (A_0,\ldots,A_{n})$ of frequencies 
$$ A_{i} = \#\{c\in\mathcal{C} : w(c)=i\} $$ 
of each possible weight $i=0,\ldots,n$.
We also consider the \textit{weight enumerator} which is the integral polynomial defined by: 
\begin{equation} \label{eq: W(x)}
W_{\CC}(x)=A_0 + A_1 x +\cdots+ A_{n}x^{n}.
\end{equation}

A linear code $\CC$ is \textit{cyclic} if for every $(c_{0},\ldots,c_{n-1})$ in $\CC$ the shifted codeword  
$(c_{1},\ldots,c_{n-1},c_{0})$ is also in $\CC$.  
If $\CC$ is a cyclic code of length $n$ over $\ff_{q}$, any word $(c_0,\ldots, c_{n-1})\in \CC$ has the associated polynomial
	$$ c_{0}+c_{1}x +c_{2}x^{2}+\cdots + c_{n-1}x^{n-1} \in \ff_{q}[x]/(x^n -1).$$
In this case, the cyclic code $\CC$ corresponds to an ideal of the ring $\ff_{q}[x]/(x^{n}-1)$ under this correspondence, say $\CC=\langle g(x) \rangle$. 
Thus, there exists a polynomial $h(x)$ such that 
	$$ g(x)h(x) = x^{n}-1.$$
The polynomial $g(x)$ is the generator polynomial of $\CC$ and the reciprocal polynomial $h^\perp(x)= h(x^{-1})x^{\deg h}$ of $h(x)$ is the check polynomial of $\CC$ and is a generator polynomial of the dual code $\CC^\perp$, which is also cyclic.
The cyclic code $\CC$ is called \textit{irreducible} (or \textit{minimal}) if the polynomial $h(x)$ 
is irreducible over $\ff_{q}$. In this case, if $\alpha$ is any root of $h(x)$, then $h(x)=m_{\alpha}(x)$, the minimal polynomial of $\alpha$ over $\ff_q$. 

We will be concerned with the weight distribution of the following $r$-ary cyclic codes. 

\begin{defi}
Let $p, s, m \in \N$ with $p$ prime, $s \mid m$, and consider
the prime powers 
$r=p^s$ and $q=p^m$.
Thus, we have the finite field extensions $ \ff_q/\,\ff_r$ and $\ff_r/\, \ff_p$.
For any $k\in  \N$ with $k\mid q-1$ define   
$$n=\tfrac{q-1}{k}.$$
For $k, q, r, n$ as above consider the $r$-ary cyclic code
\begin{equation} \label{eq: Ckqs}
	\mathcal{C}(k,q/r) = 
	\big \{ c_{\gamma} = \big( \Tr_{q/r}(\gamma\, \omega^{k i}) \big)_{i=0}^{n-1} : \gamma \in \ff_{q} \big \} \subset \ff_r^n
\end{equation} 
where $\omega$ is a primitive element of $\ff_{q}$ and $\Tr_{q/r}$ is the relative trace map from $\ff_{q}$ to $\ff_{r}$.
\end{defi}

From the classic theory of cyclic codes, these are irreducible cyclic codes over $\ff_r$ with zero $\alpha=\omega^{-k}$, of length $n$ and dimension $h$ equal to the order of $r$ modulo $n$, i.e.
	$$ h=ord_{n}(r) $$	
(i.e., the least integer $t$ such that $r^t \equiv 1 \! \pmod n$). 
That is, $\mathcal{C}(k,q/r)$ is an $[n,h]$ $r$-ary code. 

In this work, we will be concerned only with ICCs of length $2^\ell$ for any $\ell \in \N$. Since these codes were studied by Sharma, Bakshi and Raka we call them SBR-codes. They divided them in two kinds each having two cases. So, we will talk about SBR-codes of type I$a$, I$b$, II$a$ and II$b$ (see Section \ref{sec: GPs asociados a ICCs}).

The weight distribution of (irreducible) cyclic codes is in general a difficult task. 
There are several papers on the computation of the spectra of some of these codes using exponential sums. 
Baumert and McEliece were one of the first authors to compute the spectrum in terms of Gauss sums (\cite{BMc}, \cite{BMy}, \cite{Mc}). 
Gaussian sums and Gauss periods are related and Ding, in 2009, showed (\cite{Di1}, \cite{Di2}) that the weights of the irreducible cyclic codes $\CC(k,q/r)$ can be calculated in terms of Gaussian periods. Many of these results can be found in the Ding-Yang's survey \cite{DY}.

\subsubsection*{The spectral relation between codes and graphs}
Recently in \cite{PV26b}, we gave the following result showing that one can obtain the weight distribution of $\CC(k,q/r)$ in terms of the spectrum of certain integral connected GP-graph and conversely.

\begin{thm}[\cite{PV26b}, Theorem~4.3] \label{teo: rel Ckq Gkq}
Let $p$ be a prime and $q,r,k,n\in \N$ as in Definition~\ref{def: GP-graphs} such that $q=r^{h}$ with $h=ord_{n}(r)$.
Then, the weight distribution of the irreducible cyclic code $\CC(k,q/r)$ as in \eqref{eq: Ckqs} and the spectrum of the integral connected GP-graph $\G(k_r,q)$ as in \eqref{eq: Gkq}, where 	
\begin{equation} \label{eq: k_r}
	k_r = \gcd(k,\tfrac{q-1}{r-1}),	
\end{equation}
determine each other. 
More precisely, for any $\gamma\in \ff_{q}$, the eigenvalue $\lambda_{\gamma}$ of $\G(k_r,q)$ and the weight $w(c_\gamma)$ of $c_{\gamma} \in \CC(k,q/r)$ are related by 
\begin{equation} \label{Pesoaut}
	\lambda_{\gamma} = \frac{q-1}{k_r} -\frac{k r}{k_r(r-1)} w(c_{\gamma}), 
\end{equation}
and the multiplicity $m(\lambda_{\gamma})$ of $\lambda_{\gamma}$ is the frequency $A_{w(c_{\gamma})}$ of $w(c_\gamma)$ for all $\gamma\in \ff_q$. 
In particular, the multiplicity of $\lambda_0 = \frac{q-1}{k_r}$ is $A_0=1$. 
\end{thm}

\subsection*{Outlines}
The paper has the following structure. In Section \ref{sec: GPs} we give the spectrum of four families of GP-graphs: complete, Hamming, semiprimitive and Cartesian products of semiprimitive GP-graphs. 

In Section \ref{sec: GPs asociados a ICCs} we characterize the GP-graphs $\G(k,q)$ associated with SBR-codes of types I$a$, I$b$, II$a$ and II$b$. 
In Proposition \ref{prop: GPs from SBR} we give the precise parameters $k,q$ for these graphs, namely 
\begin{equation} \label{eq: 4GPs}
\G(1,r), \qquad 
\G\big(\tfrac{r+1}{2^{\ell-1}},r^2\big), \qquad 
\G\big(\tfrac{r^{2^{\ell-b}}-1}{2^{\ell-b} (r-1)},r^{2^{\ell-b}}\big), \qquad 
\G\big(\tfrac{r^{2^{\ell-b}}-1}{2^{\ell-1}(r-1)},r^{2^{\ell-b}}\big),
\end{equation}		 
for certain $b, r\in \N$.
Moreover, in Corollary~\ref{coro: the 4 GP-graphs} we show that these graphs correspond exactly to complete graphs, Hamming GP-graphs, semiprimitive GP-graphs and Cartesian products of semiprimitive GP-graphs. 

In Section \ref{sec: weight enumerators} we compute the weight enumerators for all the SBR-codes. First, in Proposition~\ref{prop: SpecGP SBR} we give the spectrum of the GP-graphs in \eqref{eq: 4GPs}. 
Then, using the spectral relation between irreducible cyclic codes and GP-graphs (see Theorem~\ref{teo: rel Ckq Gkq}), we get our main result Theorem~\ref{thm: WC(x) para SBR} giving explicit expressions for the weight enumerators $W_\mathcal{C}(x)$ for any irreducible cyclic code $\CC$ of length $2^\ell$.  
Namely, suppose that $p$ is an odd prime, $n=2^{\ell}$, $r=p^{s}$ with $\ell, s \in \N$ and $q=r^{h}$ with $h=ord_{n}(r)$. 
If 
	$$r=2^{b}c\pm 1$$ 
with $b, c \in \N$ such that $b\ge 2$ and $c$ odd, the weight enumerator of   
the $r$-ary irreducible cyclic code $\CC=\CC(\frac{q-1}{n},q/r)$ of length $2^\ell$ with $\ell\ge 2$ is given by: 

$\bullet$ If $r\equiv 1\pmod{4}$, then 
	$$ W_{\CC}(x) = \begin{cases}
					(r-1)x^{2^{\ell}} +1 	& \qquad \text{if $\ell \le b$}, \\[1mm] 
	\big( (r-1)x^{2^b} +1 \big)^{2^{\ell-b}} 	& \qquad \text{if $\ell > b$}.	
	\end{cases} $$

	$\bullet$ If $r\equiv 3\pmod{4}$, then 
	$$ 
	W_{\CC}(x) = 
	\begin{cases}
	(r^2-1) x^{2^{b+1}-2} +1		&  \text{if $c=1, \ell = b+1$}, \\[1mm]
	\big( (r^2-1) x^{2^{b+1}-2} +1\big)^{2^{\ell-b-1}}	& \text{if $c=1, \ell > b+1$}. 
	\\[1mm]	
	\big(r^{2}-1-2^{\ell-1}(r-1) \big)x^{2^{\ell}} + 2^{\ell-1}(r-1) x^{2^{\ell}-2} + 1		& \text{if $\ell < b+1$ (any $c$)}, 
	\\[1mm]
	\big(r^{2}-1-2^{\ell-1}(r-1) \big)x^{2^{\ell}} + 2^{\ell-1}(r-1) x^{2^{\ell}-2} + 1		& \text{if $c>1, \ell = b+1$}, 
	\\[1mm]
	\big( (r^{2}-1-2^b(r-1)) x^{2^{b+1}} + 2^{\ell-1}(r-1) x^{2^{\ell}-2} + 1 \big)^{2^{\ell-b-1}}	&  \text{if $c>1, \ell > b+1$}.	
\end{cases} 
$$

In Section \ref{sec: reduction}, we provide structural reductions for the SBR-codes and their associated GP-graphs. We prove in Theorem \ref{thm: SBR reduction} that each family of codes possesses a threshold minimal code $\CC_{\ell_0}$, and that the weight enumerators for all longer codes are simply exact powers of the threshold weight enumerator. 
Namely, 
	$$ W_{\CC_\ell}(x) = \big( W_{\CC_{\ell_0}}(x) \big)^{2^{\ell-\ell_0}} $$ 
where $\ell_0=b$ if $r\equiv 1 \pmod 4$ and $\ell_0=b+1$ if $r\equiv 3 \pmod 4$.

Finally, in Section \ref{sec: examples}, we illustrate these theoretical results with explicit computations. We provide the complete parameters, threshold codes, and exact weight enumerators for irreducible cyclic codes of length $2^\ell$ over finite fields of characteristics $p=3, 5,$ and $7$, summarizing the algebraic and geometric data in detailed tables.

\section{Some families of GP-graphs with known spectrum} \label{sec: GPs}
In this short preliminary section, we briefly recall the basic properties of GP-graphs and some families of GP-graphs with known spectrum.

Recall from the introduction that the \textit{generalized Paley graph} is the Cayley graph  
$$ \G(k,q) = Cay(\ff_{q},R_{k}) \qquad \text{with } \qquad R_{k} = \{ x^{k} : x \in \ff_{q}^*\}, $$
where $\ff_{q}$ is a finite field of $q$ elements, say $q=p^m$ with $p$ prime and $m \in \N$. 
That is, $\G(k,q)$ is the graph with vertex set $\ff_{q}$ and two vertices $u,v \in \ff_{q}$ are neighbors (directed edge) 
if and only if 
$v-u=x^k$ for some $x\in \ff_{q}^*$. 
The graphs $\Gamma(k,q)$ are also denoted $GP(q,\frac{q-1}k)$ in \cite{LP}.
If $\omega$ is a primitive element of $\ff_{q}$, then $R_{k} = \langle \omega^{k} \rangle = \langle \omega^{(k,q-1)} \rangle$. This implies that 
$$ \G(k,q)= \G(k',q) \qquad \text{with} \qquad k'=\gcd(k,q-1)$$
and that it is a $\frac{q-1}{\gcd(k,q-1)}$-regular graph. 
Thus, we can always assume that $k \mid q-1$. 
We have the following equivalences relative to the structure of GP-graphs: 
\begin{itemize}[leftmargin=2.5em]
\item $\G(k,q)$ is undirected if and only if $q$ is even or else $q$ is odd and $k \mid \tfrac{q-1}2$. \sk 

\item $\G(k,q)$ is connected if and only if $ord_{n}(p)=m$ where $n=\tfrac{q-1}{k}$.
\end{itemize}
and the following ones relative to the nature of the spectrum:  
	$$ Spec(\G(k,q)) \subset \mathbb{R} \quad \Leftrightarrow \quad k \mid \tfrac{q-1}2 \qquad \text{and} \qquad 
	Spec(\G(k,q)) \subset \mathbb{Z} \quad \Leftrightarrow \quad k \mid \tfrac{q-1}{p-1}.$$

\subsection*{Four families of GP-graphs with known spectrum}
The spectrum of GP-graphs is known for some few families. 
Here we list four families of GP-graphs having integral spectrum, that will be needed later.

\noindent{($a$)} \textit{Complete GP-graphs}. 
The spectrum of $\G(1,q)$ is integral and given by 
$$
Spec(\G(1,q)) = \big\{[-1]^{q-1}, [q-1]^{1}\big\}.
$$
since $\G(1,q)$ is the complete graph $K_q$.

\sk 

\noindent{($b$)} \textit{Semiprimitive GP-graphs}. 
Recall that a pair $(k,q)$ is a semiprimitive pair if $q=p^{m}$ is a prime power of $p$, such that $k\mid q-1$ and $-1$ is a power of $p$ modulo $k$. 
Equivalently, if $k\mid p^{j}+1$ for some $j\mid m$ with $m$ even.  
The spectrum of $\G(k,q)$ is integral and given by (see \cite{PV25b})
$$
Spec(\G(k,q)) = \big\{ [\tfrac{q-1}{k}]^{1},[ \tfrac{\sigma(k-1)\sqrt{q}-1}{k}]^{\frac{q-1}{k}}, [-\tfrac{\sigma\sqrt{q}+1}{k}]^{\frac{(k-1)(q-1)}{k}} \big\}, 
$$
where $\sigma=(-1)^{s+1}$ with $s=\frac{m}{2t}$ and $t$ the minimum $j$ such that $k\mid p^{j}+1$.

\sk 

\noindent{($c$)} \textit{Hamming GP-graphs}. 
A \textit{Hamming graph} $H(b,q)$ is a graph with vertex set $V=K^b$ where $K$ is any set of size $q$
(typically $\ff_q$ in applications), and where two $b$-tuples form and edge if and only if they differ in exactly one coordinate. Clearly, $H(b,q)$ is a connected graph.

Now, those GP-graphs which are Hamming graphs were characterized in \cite{LP}. 
In fact, a connected $\G(k,p^m)$ is Hamming if and only if 
	$$ k= \tfrac{p^{m}-1}n \qquad \text{with} \qquad n = b(p^{a}-1) $$ 
for some $b\mid m$ with $b>1$ and $m=ab$. In other words, the graph $\G(k,q)$ is Hamming if and only if 
	$$ k= \tfrac{p^{ab}-1}{b(p^a-1)} $$
with $b\in \N$ and $b(p^a-1)$ is a primitive divisor of $p^m-1$ (see \cite{PV21} for more details on this condition).
Moreover, in this case
\begin{equation}\label{eq: Hamming}
	\G(\tfrac{p^{ab}-1}{b(p^{a}-1)},p^{ab})\cong H(b,p^{a}).
\end{equation}

The spectra of $H(b,p^{a})$ is well-known, since 
any Hamming graph is the Cartesian product of the same complete graph of a given size, for instance 
	$$ H(b,p^{a}) \cong \square^b K_{p^{a}}.$$ 
Thus, Hamming graphs have integral spectrum given by (see \cite{PV25b})  
\begin{equation}\label{eq: Spec Hamming}
Spec(H(b,p^{a})) = \big\{ [ \ell\cdot p^{a}-b ]^{\binom{b}{\ell}(p^{a}-1)^{b-\ell}} : 0 \le \ell \le b \big\}.
\end{equation}

\noindent{($d$)} \textit{Cartesian decomposable GP-graphs}. 
Cartesian decomposable GP-graphs, that is those $\G(k,q)$ that can be decomposed as a Cartesian product for certain graphs 
$\G_1, \ldots, \G_t$, i.e.\@
$$\G(k,q) = \G_1 \,\square \cdots \square\, \G_t$$ 
were studied by Pearce and Praeger for undirected GP-graphs. 
In \cite{PP}, they showed that a connected GP-graph $\G(k,p^m)$ is Cartesian decomposable if and only if 
$$\tfrac{p^m-1}{k}=bc \qquad \text{with} \qquad b>1, \: b\mid m \: \text{ and } \: c \dagger p^{\frac mb}-1.$$ 
In this case, $\G$ turns up to be the product of a single connected GP-graph, that is 
$$\G \simeq \square^b \G_0 \qquad \text{with} \qquad \G_0 = \G \big(\tfrac{p^{\frac mb}-1}{c},p^{\frac mb}\big),$$ 
where $b$ and $c$ are integers as described above. 
The same characterization also hold for the directed case (see \cite{PV22}).

Recall that the eigenvalues of a Cartesian graph product are the sum of the eigenvalues of its factors. 
This implies that if 
	$$ Spec (\G_0) =\{[\lambda_1]^{m_1},\ldots, [\lambda_s]^{m_s}\},$$ 
then the eigenvalues of a Cartesian product graph $\G \simeq \square^b \G_0$ are given by
	$$
	\Lambda_{\ell_1,\ldots,\ell_s}=\ell_1 \lambda_1 +\cdots +\ell_s \lambda_s,
	$$
with $\ell_1,\ldots,\ell_s$ are non-negative integers such that $\ell_1+\cdots+\ell_s=b$. 
For each distinct eigenvalue $\lambda$, its multiplicity is
$$
m(\lambda)
=
\sum_{\substack{
		\ell_1,\ldots,\ell_s\geq 0, \\
		\ell_1+\cdots+\ell_s=b \\
		\Lambda_{\ell_1,\ldots,\ell_s} = \lambda
}}
	\tbinom{b}{\ell_1,\ldots,\ell_s} \, m_1^{\ell_1}\cdots m_s^{\ell_s}.
$$

Observe that the Hamming graph is a particular case of a Cartesian decomposable GP-graph.

It is worth noticing that, besides the families listed above, only $\G(3,q)$ and $\G(4,q)$, recently the eleven GP-graphs (see \cite{PV25b}) coming from the eleven exceptional 2-weight irreducible cyclic codes obtained by Schmidt and White (\cite{SW}), and the bilinear form GP-graphs (see \cite{PV26c}) have known spectrum.

\section{Characterizing the GP-graphs associated to SBR-codes} \label{sec: GPs asociados a ICCs}
In 2007, Sharma, Bakshi and Raka obtained the weight distribution of all $q$-ary codes of length $2^m$ and $q$ odd (see \cite{SBR}). We will refer to them as the \textit{Sharma-Bakshi-Raka codes} or simply by \textit{SBR-codes}.
They divided these codes into two families, according to whether $q$ is of the form 
$$ 2^bc +1 \qquad \text{or} \qquad 2^bc -1 \qquad \text{with} \quad b\ge 2, \: c \text{ odd},$$ 
and in each case they distinguished two possibilities:
$$ m \le b \qquad \text{or} \qquad m\ge b+1.$$

Here, we will first characterize all GP-graphs which are spectrally associated to these SBR-codes, which turn up to be of four different kinds: 
\begin{itemize}
\item complete graphs, \sk 

\item semiprimitive GP-graphs, \sk 

\item Hamming GP-graphs, and \sk

\item Cartesian product of semiprimitive GP-graphs. 
\end{itemize}
All these graphs have known spectrum (see Section \ref{sec: GPs}). 
Using the results in the previous sections we re-obtain the weight enumerators of these codes from the spectrum of the associated graphs. In particular, we get closed expressions for $W_\CC(x)$ in all the cases. 

\sk 

Since for us $q=p^m$, and we are also interested in $r$-ary codes, we will use the results in \cite{SBR} for $r$-ary codes of length $2^\ell$. 
Also, notice that if $r$ is odd, then $r\equiv \pm 1 \pmod{4}$, and hence we can assume that $r=2^{b} c \pm 1$ for some $b\in \N$ with $b\ge 2$ and $c$ odd.

From now on we will use the following notation:
	\begin{equation} \label{eq: k_{l,r}}
		k_{\ell,r} = \gcd(\tfrac{q-1}{2^\ell},\tfrac{q-1}{r-1}).
	\end{equation}
That is, $k_{\ell,r}$ is the number $k_r$ defined in \eqref{eq: k_r} taking $k=\frac{q-1}{2^\ell}$.

We want to give all GP-graphs $\G(k_{\ell,r},q)$ associated to the SBR-codes $\CC(k,q/r)$ of length $2^\ell$.
First, we compute the possible parameters $q$ and $k_r$ for these graphs.

\begin{lem} \label{lem: kr SBR}
Let $p$ be an odd prime and take $r=p^{s}$ and $q=r^{h}$ with $h=ord_{n}(r)$ where $n=2^{\ell}$ and $\ell, s \in \N$. 
Let $k_{\ell,r}$ 
as in \eqref{eq: k_{l,r}} and assume that $r=2^{b}c\pm 1$ with $b,c\in \N$, $b\ge 2$ and $c$ odd. 
If $\ell=1$, then $q=r$ and $k_{\ell,r}=1$ while if $\ell>1$ we have  
\begin{equation} \label{eq: q SBR}
	q=
	\begin{cases}
	\hfil	r 				& \qquad \text{if $2 \le \ell \le b$ and $r\equiv 1 \!\!\pmod{4}$,} \\[1mm]
	\hfil	r^2				& \qquad \text{if $\ell =b+1$ and $r\equiv 1 \!\!\pmod{4}$,} \\[1mm]
	\hfil	r^2 			& \qquad \text{if $2\le \ell \le b+1$ and $r\equiv 3 \!\!\pmod{4}$,} \\[1mm]
		r^{2^{\ell-b}} 	& \qquad \text{if $\ell > b+1$ \,$($any $r)$,} 
	\end{cases}
\end{equation}
and 
\begin{equation} \label{eq: kr SBR}
	k_{\ell,r}=
	\begin{cases}
		\hfil 1 										& \qquad \text{if $r\equiv 1 \!\!\pmod{4}$ and $2 \le \ell \le b$,} \\[2mm]
		\hfil \frac{r^{2^{\ell-b}}-1}{2^{\ell-b}(r-1)}  & \qquad \text{if $r\equiv 1 \!\!\pmod{4}$ and $\ell > b$,} \\[2mm]
		\hfil \frac{r+1}{2^{\ell-1}} 				    & \qquad \text{if $r\equiv 3 \!\!\pmod{4}$ and $2\le \ell \le b+1$,} \\[2mm]
		\hfil \frac{r^{2^{\ell-b}}-1}{2^{\ell-1}(r-1)}  & \qquad \text{if $r\equiv 3 \!\!\pmod{4}$ and $\ell > b+1$.}
	\end{cases}
\end{equation}
\end{lem}

\begin{proof}
As the authors mentioned in \cite{SBR} (see equation (1) in page 1087), we have that
\begin{equation}\label{eq: ord SBR}
	ord_{n}(r)=
	\begin{cases}
	\hfil	1 & \qquad \text{if $\ell=1$ or else if $r\equiv 1 \!\!\pmod{4}$ and $2 \le \ell \le b$,} \\
	\hfil	2 & \qquad \text{if $r\equiv 1 \!\!\pmod 4$ and $\ell = b+1$,}\\
	\hfil	2 & \qquad \text{if $r\equiv 3 \!\!\pmod 4$ and $2\le \ell \le b+1$,}\\
		2^{\ell-b} & \qquad \text{if $\ell > b+1$.} 
	\end{cases}
\end{equation}
Thus, we obtain the following values of $q$ in these cases
$$
q=	\begin{cases}
\hfil r 				& \qquad \text{if $\ell=1$ or else if $r\equiv 1 \!\!\pmod{4}$ and $2 \le \ell \le b$,} \\
\hfil	r^2				& \qquad \text{if $\ell =b+1$ and $r\equiv 1 \!\!\pmod{4}$,} \\
\hfil	r^2 			& \qquad \text{if $r\equiv 3 \!\!\pmod{4}$ and $2\le \ell \le b+1$,}\\
	r^{2^{\ell-b}}  & \qquad \text{if $\ell > b+1$ \, $($any $r)$, } 
\end{cases}
$$
and thus \eqref{eq: q SBR} follows.

We now prove \eqref{eq: kr SBR}. 
Recall that 
$\gcd(\tfrac{d}{a}, \tfrac{d}{b}) = \tfrac{d}{\lcm(a,b)}$
for any $a, b \in \Div(d)$ with $d \in \N$.
The assertion is clear when $\ell=1$ or else if $1<\ell \le b$ and $r=1+2^{b}c$, since in this case $q=r$ and so  $k_{\ell,r} = \gcd(\frac{q-1}{2^\ell},\frac{q-1}{r-1})=1$.
When $\ell>b$ and $r=1+2^b c$ with $c$ odd then $\lcm(n,r-1) = 2^{\ell}c = 2^{\ell-b}(r-1)$, which implies that 
$$
	k_{\ell,r} = \frac{r^{2^{\ell-b}}-1}{2^{\ell-b}(r-1)}.
$$ 
On the other hand, if $r=2^{b}c-1$ and $1<\ell$, then 
$\lcm(n,r-1) = \tfrac{n(r-1)}{2} = 2^{\ell-1}(r-1)$. 
This implies that
$$	k_{\ell,r}=
\begin{cases}
	\frac{r^{2}-1}{2^{\ell-1}(r-1)} 			& \qquad \text{if $2\le \ell\le b+1$,} \\[2mm]
	\frac{r^{2^{\ell-b}}-1}{2^{\ell-1}(r-1)} 	& \qquad \text{if $\ell>b+1$,}
\end{cases} $$
from which equation \eqref{eq: kr SBR} follows.
\end{proof}

Using the previous result we are able to obtain the GP-graphs corresponding to all $r$-ary irreducible cyclic codes of length $2^{\ell}$. 
We will show that they belong to four different kinds: complete, semiprimitive, Hamming and Cartesian decomposable GP-graphs.

\begin{prop} \label{prop: GPs from SBR}
Let $n=2^{\ell}$ with $\ell \in \N$, $p$ an odd prime and let $r=p^{s}$ and $q=r^{h}$ with $h=ord_{n}(r)$. 
Let $k_{\ell,r}$ as in \eqref{eq: k_{l,r}} 
and let $r=2^{b}c\pm 1$ with $b, c \in \N$ such that $b\ge 2$ and $c$ odd. 
Then, the graph $\G=\G(k_{\ell,r},q)$ takes the different possible forms: if $\ell=1$ then $\G=\G(1,r)$ and if $\ell \ge 2$ we have 
\begin{equation} \label{eq: Gkrq}
	\G=\begin{cases}
		\G(1,r) 						& \qquad \text{if $r\equiv 1\!\!\!\!\pmod 4$ and $\ell \le b$}, \\[1mm]
		\G(\frac{r^{2^{\ell-b}}-1}{2^{\ell-b} (r-1)},r^{2^{\ell-b}})	& \qquad \text{if $r\equiv 1\!\!\!\!\pmod 4$ and $\ell > b$}, \\[1mm]
		\G(\frac{r+1}{2^{\ell-1}},r^2)	& \qquad \text{if $r\equiv 3\!\!\!\!\pmod{4}$ and $\ell \le b+1$}, \\[1mm]
		\G(\frac{r^{2^{\ell-b}}-1}{2^{\ell-1}(r-1)},r^{2^{\ell-b}})	& \qquad \text{if $r\equiv 3\!\!\!\!\pmod{4}$ and $\ell > b+1$}. 
	\end{cases}
\end{equation}
\end{prop}

\begin{proof}
Let us denote by ($a$), ($b$), ($c$) and ($d$) the four cases in \eqref{eq: Gkrq}.

\noindent ($a$) 
This is clear from Lemma \ref{lem: kr SBR}, since in this case $k_r=1$ and $q=r$. 

\noindent $(b)$ 
By direct computation, in this case we have that $n_r=2^{\ell-b}(r-1)$ and $m=2^{\ell-b}$. 

\noindent ($c$) 
This also follows from the Lemma \ref{lem: kr SBR} since in this case $k_r=\frac{r+1}{2^{\ell-1}}$ and $q=r^{2}$. 

\noindent $(d)$ 
Finally, when $\ell>b+1$ and $r=2^{b}c-1$, by direct computation
	$$	
		k_{\ell,r} = \frac{r^{2^{\ell-b}}-1}{2^{\ell-1}(r-1)} \qquad \text{and} \qquad q=r^{2^{\ell-b}},
	$$ 
as asserted.
\end{proof}

As a direct consequence, under the same hypothesis and notations of Proposition \ref{prop: GPs from SBR} we have the following. 
\begin{coro} \label{coro: the 4 GP-graphs}
Let $\G=\G(k_{\ell,r},q)$ be one of the graphs in \eqref{eq: Gkrq}. 
In the first and second cases, $\G$ is the complete graph $K_r$ and the Hamming graph $H(2^{\ell-b},r)$, respectively. 
In the third and fourth cases, if $c>1$ or else if $c=1$ and $\ell<b+1$, then $\G$ is respectively a semiprimitive GP-graph or 
Cartesian decomposable GP-graph (i.e.,  $\G \simeq \square^{2^{\ell-b-1}} \, \G(\tfrac{r+1}{2^{b}},r^2)$ is the Cartesian product of isomorphic copies of a semiprimitive GP-graph). 
Moreover, if $c=1$ and $\ell=b+1$, in the third case $\G$ is the complete graph $K_{r^2}$ and in the fourth one if $c=1$ and $\ell>b+1$, $\G$ is a Hamming graph, more precisely $\G=H(2^{\ell-b-1},r^{2})$.
\end{coro}	 

\begin{proof}
By Proposition \ref{prop: GPs from SBR}, the first case is when $k_{\ell,r}=1$ and $q=r$ and so $\G(k_{\ell,r},q)$ is the complete graph $K_r$. 

In the second case, we have that $n_r=2^{\ell-b}(r-1)$ and $m=2^{\ell-b}$. By Lemma~4.2 from \cite{PV26b}, we have that $\frac{q-1}{k_r}$ is a primitive divisor of $q-1$ and so $\G(k_{\ell,r},q)$ is a connected GP-graph, thus by \eqref{eq: Hamming} we obtain that $\G(k_{\ell,r},q)$ is a Hamming GP-graph isomorphic to $H(2^{\ell-b},r)$, as asserted.

For the third case. If $c>1$ or $\ell<b+1$, then $k_{\ell,r} = 2^{b-\ell+1}c$ with $p^s=r\equiv-1\pmod{k_{\ell,r}}$ and hence the graph is semiprimitive.
If $c=1$ and $\ell=b+1$, then $k_r=1$, so the graph is the complete graph $K_{r^2}$.

Finally, in the fourth case, $\ell>b+1$ and $r=2^{b}c-1$. 
Put $\ell=b+t$ for some integer $t\ge 2$. Now, notice that $r^{2^t}=(r^{2})^{2^{t-1}}$, which implies that 
$$	
\frac{r^{2^{\ell-b}}-1}{2^{\ell-1}(r-1)} = \frac{r^{2^t}-1}{2^{b}2^{t-1}(r-1)} = \frac{(r+1)(r^{2^t}-1)}{2^{b}2^{t-1}(r^2-1)}.
$$ 
By Lemma~4.2 from \cite{PV26b} the graph $\G(k_r,q)$ is connected and so the last GP-graph of the list is the following Cartesian-decomposable graph (see ($d$) of Section \ref{sec: GPs})
$$
\G(\tfrac{r^{2^t}-1}{2^{b}2^{t-1}(r-1)},r^{2^t})\simeq \square^{2^{t-1}} \, \G(\tfrac{r^2-1}{2^b(r-1)},r^2)= \square^{2^{t-1}} \, \G(\tfrac{r+1}{2^b},r^2).
$$
Notice that if $c=1$, then we obtain that $\G$ is the product of $2^{\ell-b-1}$-copies of the complete graph $K_{r^2}$. On the other hand if $c>1$, we obtain that $\G$ is the the product of $2^{\ell-b-1}$-copies of the semiprimitive GP-graph $\G(\tfrac{r+1}{2^b},r^2)$, as desired.
\end{proof}

\section{The weight enumerators of SBR-codes} \label{sec: weight enumerators}
We begin by giving names to the four cases of SBR-codes adapted to our $r$-ary notation. Let $\CC$ be an $r$-ary irreducible cyclic code of length $2^\ell$, where $r = 2^b c \pm 1$ with $b \ge 2$ and $c$ odd.
	\begin{itemize}[leftmargin=2.5em]
		\item  \textbf{Type I codes}: if $r=2^bc +1$. We call $\CC$ of \textbf{type Ia} if $\ell \le b$ and of \textbf{type Ib} if $\ell > b$. \msk
		
		\item \textbf{Type II codes}: if $r=2^bc -1$. We call $\CC$ of \textbf{type IIa} if $\ell \le b+1$ and of \textbf{type IIb} if $\ell > b+1$.
	\end{itemize}

We point out that for the SBR-codes of type II, except for the trivial cases of lengths $n=2,4$ (i.e., $\ell=1,2$), the weight enumerator of the codes are not given explicitly in \cite{SBR}, but with a somewhat complicated expression.

Here, we will obtain explicit expressions of the weight enumerators for all the SBR-codes. 
We will first obtain the spectrum of the associated GP-graphs obtained in the previous section and then use the spectral relation between GP-graphs and ICC codes in this case.

We begin by giving the spectrum of the GP-graphs associated to SBR-codes.

\begin{prop} \label{prop: SpecGP SBR}
Let $n=2^{\ell}$ with $\ell \in \N$, $p$ an odd prime and let $r=p^{s}$ with $s \in \N$ and take $q=r^{h}$ with $h=ord_{n}(r)$. 
Let $k_{\ell,r} = \gcd(\frac{q-1}{2^\ell},\frac{q-1}{r-1})$ and let $r=2^{b}c\pm 1$ with $b,c \in \N$ such that $b\ge 2$ and $c$ odd. Then, the spectrum of the graphs $\G(k_{\ell,r},q)$ are given by

\noindent $(a)$ Assume $r\equiv 1\pmod{4}$. 

\begin{enumerate}[$(i)$]
	\item If $\ell \le b$, then 
$$ Spec(\G(k_{\ell,r},q)) = \big\{[-1]^{r-1}, [r-1]^{1}\big\}.$$ 

	\item If $\ell > b$, then  
$$ Spec(\G(k_{\ell,r},q)) = \big\{ [2^{\ell-b}(r-1)-jr]^{\binom{2^{\ell-b}}{j}(r-1)^j} \big\}_{\{0\le j \le 2^{\ell-b}\}}. $$ 
\end{enumerate}

\noindent $(b)$ Assume $r\equiv 3\pmod{4}$. 

\begin{enumerate}[$(i)$]	
	
	\item If $\ell=b+1$ and $c=1$, then 
	\begin{equation*} \label{eq: Spec pot 2 comp}
		Spec(\G(k_{\ell,r},q)) = \big\{[r^2-1]^{1}, [-1]^{r^2-1}\big\}.
	\end{equation*}
	
	\item If $\ell > b+1$ and $c=1$, then 
$$ Spec(\G(k_{\ell,r},q)) = \big\{ [2^{\ell-b-1}(r^2-1)-jr^2]^{\binom{2^{\ell-b-1}}{j}(r^2-1)^j} \big\}_{\{0\le j \le 2^{\ell-b-1}\}}. $$ 

	\item If $2 \le \ell < b+1$ or else 
	$\ell=b+1$ and $c>1$, then 
	\begin{equation*} \label{eq: Spec pot 2 semip}
		Spec(\G(k_{\ell,r},q)) = \big\{[2^{\ell-1}(r-1)]^{1}, [r-2^{\ell-1}]^{2^{\ell-1}(r-1)}, [-2^{\ell-1}]^{r^{2}-1-2^{\ell-1}(r-1)}\big\}.
	\end{equation*}
	
	\item If $\ell > b+1$ and $c>1$, then
		$$ Spec(\G(k_{\ell,r},q)) = \big\{ [\Lambda_{\ell_1,\ell_2,\ell_3}]^{\mu_{\ell_1,\ell_2,\ell_3}} \big\}, $$
	where the eigenvalues are given by
		$$ \Lambda_{\ell_1,\ell_2,\ell_3} = \ell_1 (2^{b}(r-1))+ \ell_2 (r-2^{b}) +\ell_3 (-2^{b}) $$
	and their corresponding multiplicities are
		$$ \mu_{\ell_1,\ell_2,\ell_3} = \tbinom{2^{\ell-b-1}}{\ell_1,\ell_2,\ell_3}m_1^{\ell_1} m_2^{\ell_2} m_3^{\ell_3} $$
	for all non-negative integers $\ell_1, \ell_2, \ell_3$ satisfying 
			$$\ell_1+\ell_2+\ell_3=2^{\ell-b-1},$$ 
	where $m_1=1$, $m_2=2^{b}(r-1)$ and $m_3=r^{2}-1-2^{b}(r-1)$. Note that some eigenvalues $\Lambda_{\ell_1,\ell_2,\ell_3}$ and some multiplicities $\mu_{\ell_1,\ell_2,\ell_3}$ might coincide for different choices of the tuples $(\ell_1, \ell_2, \ell_3)$. More precisely, for each eigenvalue $\lambda$ of $\G$ its multiplicity is given by
		$$
		m(\lambda) = \sum_{\substack{
				\ell_1,\ell_2,\ell_3\geq 0\\
				\ell_1+\ell_2+\ell_3=2^{\ell-b-1}\\
				\Lambda_{\ell_1,\ell_2,\ell_3}=\lambda
		}} 	\mu_{\ell_1,\ell_2,\ell_3}.
		$$
\end{enumerate}
\end{prop}

\begin{proof}
($a$) Item ($i$) is clear from the above proposition, since in this case $\G(k_{\ell,r},q) = K_r$ which has spectrum as asserted.
For ($ii$), by Corollary \ref{coro: the 4 GP-graphs}, the graph  $\G(k_{\ell,r},q)$ is isomorphic $H(2^{\ell-b},r)$, thus we can obtain the spectra of $\G(k_{\ell,r},q)$ directly from \eqref{eq: Spec Hamming}.

\noindent ($b$) 
Applying Corollary \ref{coro: the 4 GP-graphs} we have that: in case ($i$), $\G(k_{\ell,r},q)$ is the complete graph $K_{r^2}$ which has spectrum as asserted; in case ($ii$), $\G(k_{\ell,r},q)$ is the semiprimitive GP-graph $\G(\frac{r+1}{2^{\ell-1}})$ 
with $\sigma=1$ (see Section~\ref{sec: GPs}), so its spectrum is as stated; and in case ($iii$), 
$\G(k_{\ell,r},q)$ is isomorphic $H(2^{\ell-b-1},r^2)$, thus we can obtain the spectrum of $\G(k_{\ell,r},q)$ from \eqref{eq: Spec Hamming}.
Finally, we prove ($iv)$.
Let $t=\ell-b$. By Corollary \ref{coro: the 4 GP-graphs}, we have that
	$$
		\G(k_{\ell,r},q) \simeq \square^{2^{t-1}} \, \G(\tfrac{r+1}{2^b},r^2).
	$$
Notice that the graph $\G(\tfrac{r+1}{2^b},r^2)$ is a semiprimitive GP-graph with $\sigma=1$. 
Thus, after some straightforward computation, its spectra is given by
	$$
		Spec(\G(\tfrac{r+1}{2^{b}},r^{2})) = \{[2^{b}(r-1)]^{1}, [r-2^{b}]^{2^{b}(r-1)}, [-2^{b}]^{r^{2}-1-2^{b}(r-1)}\}.
	$$
Therefore, the last assertion of item ($f$) is a consequence of the fact that the eigenvalues of a Cartesian graph product are the sum of the eigenvalues of its factors.
\end{proof}

Now, using the spectrum of the GP-graphs just computed we can give the weight enumerators of the SBR-codes. 

\begin{thm} \label{thm: WC(x) para SBR}
Let $p$ be an odd prime, $n=2^{\ell}$, $r=p^{s}$ with $\ell, s \in \N$ and take $q=r^{h}$ with $h=ord_{n}(r)$. 
Let $r=2^{b}c\pm 1$ with $b, c \in \N$ such that $b\ge 2$ and $c$ odd. 
Consider the $r$-ary irreducible cyclic code 
	$$ \CC = \CC_{\ell,r} = \CC(\tfrac{q-1}{n},q/r). $$ 
Then, we have the following:
\begin{enumerate}[$(a)$]
	\item If $r\equiv 1\pmod{4}$ and $\ell \le b$, then $\CC$ is $[2^{\ell}, 1, 2^{\ell}]_r$-code with weight enumerator
	$$ W_{\CC}(x)=1+(r-1)x^{2^{\ell}}. $$ 
	
	\item If $r\equiv 1\pmod{4}$ and $\ell > b$, then $\CC$ is a $[2^{\ell},2^{\ell-b}, 2^b]_r$-code with weight enumerator
	$$
		W_{\CC}(x) = \big(1+(r-1)x^{2^b}\big)^{2^{\ell-b}}.
	$$

	\item If $r\equiv 3\pmod{4}$ with $\ell =b+1$ and $c=1$, then $\CC$ is a $[2^{\ell},2, 2^{b+1}-2]_r$-code with weight enumerator 
	$$
		W_{\CC}(x) = 1 + (r^2-1) x^{2^{b+1}-2}.
	$$
	
	\item If $r\equiv 3\pmod{4}$ with $\ell < b+1$ or else with 
	$\ell = b+1$ and $c>1$, then $\CC$ is a $[2^{\ell},2, 2^{\ell}-2]_r$-code with weight enumerator 
	$$
		W_{\CC}(x) = 1 + 2^{\ell-1}(r-1) x^{2^{\ell}-2}+ \big(r^{2}-1-2^{\ell-1}(r-1) \big)x^{2^{\ell}}.
	$$
	
	\item If $r\equiv 3\pmod{4}$ with $\ell > b+1$ and $c=1$, then $\CC$ is a $[2^{\ell},2^{\ell-b},2^{b+1}-2]_r$-code having weight enumerator
	$$
		W_{\CC}(x) = \big( 1+ (r^2-1) x^{2^{b+1}-2} \big)^{2^{\ell-b-1}}.
	$$
	
	\item If $r\equiv 3\pmod{4}$ with $\ell > b+1$ and $c>1$, then $\CC$ is a $[2^{\ell},2^{\ell-b},2^{b+1}-2]_r$-code having weight enumerator
	$$
	W_{\CC}(x) = \big( 1 + 2^b(r-1) x^{2^{b+1}-2} + (r^{2}-1-2^b(r-1)) x^{2^{b+1}} \big)^{2^{\ell-b-1}}.
	$$
\end{enumerate}
In particular, all non-trivial weights have frequencies that are multiples of $(r-1)$.
\end{thm}

\begin{proof}
Items ($a$)--($e$) follow directly from Theorem \ref{teo: rel Ckq Gkq} and Proposition \ref{prop: SpecGP SBR}.

Finally, to see item $(f)$, recall that by Theorem \ref{teo: rel Ckq Gkq} the weights of $\CC$ are given by
	$$
		 \tfrac{r-1}{r} \big( n-\tfrac{k_{\ell,r}}{k} \lambda \big)
	$$
with $\lambda$ an eigenvalue of $\G(k_{\ell,r},q)$.
In the notation of the spectrum of $\G(k_{\ell,r},q)$, notice that if $\ell_1+\ell_2+\ell_3 = 2^{\ell-b-1}$ then 
	$$
		n=2^{\ell}=2^{b+1}(\ell_1+\ell_2+\ell_3) \qquad \text{and} \qquad \tfrac{k_{\ell,r}}{k} = \tfrac{2}{r-1}.
	$$
Therefore, from ($f$) of Proposition \ref{prop: SpecGP SBR} and Theorem \ref{teo: rel Ckq Gkq} the weights of $\CC$ are given by
	$$
		\ell_1(\tfrac{r-1}{r})(2^{b+1}-\tfrac{2^{b+1}(r-1)}{r-1})+\ell_2(\tfrac{r-1}{r})(2^{b+1}-\tfrac{2(r-2^{b})}{r-1})+\ell_3(\tfrac{r-1}{r})(2^{b+1}+\tfrac{2^{b+1}}{r-1})
	$$
which simplifies to
	$$
		\ell_1\cdot 0+\ell_2 \cdot (2^{b+1}-2)+\ell_3\cdot 2^{b+1}
	$$
with frequencies
	$$
		\mu_{\ell_1,\ell_2,\ell_3} = \tbinom{2^{\ell-b-1}}{\ell_1,\ell_2,\ell_3}m_1^{\ell_1} m_2^{\ell_2} m_3^{\ell_3} 
	$$
where $m_1=1$, $m_2=2^{b}(r-1)$, $m_3=r^{2}-1-2^{b}(r-1)$ and $\binom{2^{\ell-b-1}}{\ell_1,\ell_2,\ell_3}=\frac{2^{\ell-b-1}!}{\ell_1 ! \ell_2 ! \ell_3 !}$. 
Therefore, the weight enumerator of $\CC$ is given by
	$$
		W_{\CC}(x) = \sum \tbinom{2^{\ell-b-1}}{\ell_1,\ell_2,\ell_3}m_1^{\ell_1} m_2^{\ell_2} m_3^{\ell_3}x^{\ell_1 \cdot0} x^{\ell_2 (2^{b+1}-2)} x^{\ell_3 2^{b+1}}
	$$ 
where the sum is taken of the $3$-tuples of non-negative integers $(\ell_1,\ell_2,\ell_3)$ such that $\ell_1+\ell_2+\ell_3=2^{\ell-b-1}$. 
By the binomial theorem we have that 
\begin{align*}
	W_{\CC}(x)  & = \big( m_1 x^{0}+m_2 x^{2^{b+1}-2} + m_3 x^{2^{b+1}}\big)^{2^{\ell-b-1}} \\ 
	& = \big( 1+2^{b}(r-1)x^{2^{b+1}-2}+ (r^{2}-1-2^{b}(r-1))x^{2^{b+1}}\big)^{2^{\ell-b-1}},
\end{align*}
as desired.

Notice that the length of all of the codes is $n=2^{\ell}$ and its dimension is $ord_{n}(r)$ which is explicitly given by \eqref{eq: ord SBR} in each case. 
Finally the minimal weight in each case can be derived directly from the weight enumerators.
The last assertion is clear from the expressions for the weight enumerators (this is implicit in Theorems 1--3 in \cite{SBR}).
\end{proof}

Note that if $\ell=1$, we are in case ($a$), and $\CC$ is the repetition code $Rep_{r}(2)$ with parameters $[2,1,2]$ and weight enumerator $W_\CC(x)=1+(r-1)x^2$.

\begin{rem}
The weight distribution of these codes were previously obtained in Theorems 1, 2 and 3 of \cite{SBR}. 
However, in the case of $r=2^bc-1$, the weight enumerators of ICCs of length $2^\ell$ with $\ell \ge 3$ are not explicitly given in a closed form (see Theorem~3 in \cite{SBR}). 
We point out that we have obtained closed expressions for $W_\CC(x)$ in all the cases. 
Also, it can be checked that the parameters of the codes given in the above result coincide with the parameters given in the final corollary in \cite{SBR}.
\end{rem}

\section{Reductions for the weight enumerators} \label{sec: reduction}
In this section, we show that each family of $r$-ary SBR-codes of length $2^\ell$ possesses a threshold minimal code $\mathcal{C}_{\ell_0}$ of length $2^{\ell_0}$. Remarkably, all subsequent codes in the family corresponding to lengths $2^\ell$ (with $\ell > \ell_0$) have weight enumerators that are exact powers of the threshold weight enumerator. 

Define the threshold parameter $\ell_0$ by
\begin{equation} \label{eq: threshold}
\ell_0 =  \begin{cases} \hfil b & \qquad \text{if } r \equiv 1 \pmod 4, \\[1mm] b+1 & \qquad \text{if } r \equiv 3 \pmod 4. \end{cases}
\end{equation}	

The reduction result for the weight enumerators is as follows.

\begin{thm} \label{thm: SBR reduction}
Let $p$ be an odd prime, $r=p^s$, and let $\CC_\ell$ denote the $r$-ary irreducible cyclic code of length $2^\ell$ for $\ell \ge 2$. 
Let $b \ge 2$ and $c$ odd be the unique integers such that $r = 2^b c \pm 1$. 
Then, the base threshold code $\CC_{\ell_0}$ has dimension $1$ if $r \equiv 1 \pmod 4$, and dimension $2$ if $r \equiv 3 \pmod 4$. 
Furthermore, for all $\ell \ge \ell_0$, the weight enumerator of $\CC_\ell$ is completely determined by the base code at the threshold $\ell_0$. 
Specifically,
\begin{equation} \label{eq: W(x) threshold}
	W_{\CC_\ell}(x) = \big( W_{\CC_{\ell_0}}(x) \big)^{2^{\ell-\ell_0}}.
\end{equation}
\end{thm}

\begin{proof}
We proceed by analyzing the parity of $r$ modulo $4$.

$\bullet$ 
Assume first that $r \equiv 1 \pmod 4$. By definition, $r = 2^b c + 1$ and the threshold is $\ell_0 = b$. 
The dimension of $\CC_b$ is $h = \operatorname{ord}_{2^b}(r)$. Since $r \equiv 1 \pmod{2^b}$, we have $h=1$.

From Theorem \ref{thm: WC(x) para SBR} $(a)$, the weight enumerator of the base code at the threshold $\ell = b$ is given by
$$W_{\CC_b}(x) = 1 + (r-1)x^{2^b}.$$
For any $\ell > b$, Theorem \ref{thm: WC(x) para SBR} $(b)$ establishes that the weight enumerator of $\CC_\ell$ is
$$W_{\CC_\ell}(x) = \big( 1 + (r-1)x^{2^b} \big)^{2^{\ell-b}}.$$
Substituting $W_{\CC_b}(x)$ into this equation directly yields $W_{\CC_\ell}(x) = (W_{\CC_b}(x))^{2^{\ell-b}}$, satisfying the claim.

$\bullet$
Assume now that $r \equiv 3 \pmod 4$. In this case, $r = 2^b c - 1$ and the threshold is $\ell_0 = b+1$. 
The dimension of $\CC_{b+1}$ is $h = \operatorname{ord}_{2^{b+1}}(r)$. On the one hand we have $r \not\equiv 1 \pmod{2^{b+1}}$ but on the other hand it holds $r^2 = 2^{2b}c^2 - 2^{b+1}c + 1 \equiv 1 \pmod{2^{b+1}}$, since $b \ge 2$, then we have $h=2$.

From Theorem \ref{thm: WC(x) para SBR} $(c)$ and $(d)$, the weight enumerator at $\ell = b+1$ is 
	$$ W_{\CC_{b+1}}(x) = 
		\begin{cases}
		1 + 2^b(r-1) x^{2^{b+1}-2} & \text{if } c=1, \\[1mm]
		1 + 2^b(r-1) x^{2^{b+1}-2} + \big( r^2-1-2^b(r-1) \big)x^{2^{b+1}} & \text{if } c>1. 
		\end{cases}
	$$
For any $\ell > b+1$, Theorem \ref{thm: WC(x) para SBR} $(e)$ and $(f)$ states that
	$$ W_{\CC_\ell}(x) = 
		\begin{cases}
		\big( 1 + 2^b(r-1) x^{2^{b+1}-2}  \big)^{2^{\ell-(b+1)}} & \text{if } c=1, \\[1mm]
		\big( 1 + 2^b(r-1) x^{2^{b+1}-2} + (r^2-1-2^b(r-1)) x^{2^{b+1}} \big)^{2^{\ell-(b+1)}} & \text{if } c>1. 
		\end{cases}
	$$
Replacing the base polynomial with $W_{\CC_{b+1}}(x)$ completes the proof.
\end{proof}

The reduction formula in Theorem \ref{thm: SBR reduction} reveals striking structural consequences for the SBR-codes beyond the threshold $\ell_0$, completely determining their minimum distances and their number of distinct non-zero weights.
	
\begin{rem}
	\label{rem: minimum distance}
As the code length $2^\ell$ grows, the minimum distance freezes exactly at the threshold value. For Type I codes, the minimum distance remains strictly $d = 2^b$ for all $\ell \ge b$. Similarly, for Type II codes, the minimum distance remains fixed at $d = 2^{b+1}-2 = 2(2^b-1)$ for all $\ell \ge b+1$.
\end{rem}
	
We now study the number of weights.
\begin{coro} \label{coro: t-weights}
	Let $\ell > \ell_0$ with $\ell_0$ as in \eqref{eq: threshold}. In the previous notations we have: 
		
	\noindent $(a)$ Type I codes are $t$-weight codes with $t = 2^{\ell-b}$. 
		
	\noindent $(b)$ Type II codes are $t$-weight codes with $t \le \frac{1}{2}m(m+3)$ where $m = 2^{\ell-b-1}$. 
	Furthermore, if $c=1$, then $t=2^{\ell-b-1}$; and if $c>1$ and $b+2\leq\ell\leq2b$, then $t=2^{\ell-b-2}\bigl(2^{\ell-b-1}+3\bigr)$.
\end{coro}
	
\begin{proof}
$(a)$ For Type I codes with $\ell > b$, the weight enumerator is given by $(1+(r-1)x^{2^b})^{2^{\ell-b}}$. This binomial power expands into exactly $2^{\ell-b}+1$ terms, yielding $t = 2^{\ell-b}$ non-zero weights.

$(b)$ For a Type II code with $\ell > b+1$, notice that if $c=1$ then the weight enumerator is given by $(1+(r^2-1)x^{2^{b+1}-2})^{2^{\ell-b-1}}$, 
the same reasoning as in item ($a$) allows us to conclude that $t=2^{\ell-b-1}$. 
In particular, the bound $t \le \frac{1}{2}m(m+3)$ holds immediately.

On the other hand, if $c>1$ the weight enumerator is the $m$-th power of the trinomial 
	$$ T(x) = 1 + A x^{2^{b+1}-2} + B x^{2^{b+1}} $$ 
established in Theorem \ref{thm: WC(x) para SBR} $(f)$. 
Since $c>1$, the coefficients are positive:
$$
A=2^b(r-1) >0 \qquad\text{and}\qquad B=2^b(c-1)(r-1)>0.
$$
Thus, the distinct weights occurring in
$T(x)^m$ are precisely the different integers
$$
w_{i,j}=(2^{b+1}-2)i+2^{b+1}j,
\qquad i,j\geq 0,\quad i+j\leq m.
$$
There are $\binom{m+2}{2}$ admissible pairs $(i,j)$.
Since only $(0,0)$ gives the zero weight, it follows that
$$
t\leq \tbinom{m+2}{2}-1=\tfrac{m(m+3)}{2},
$$
regardless of possible coincidences among the exponents.

Now, suppose that $w_{i,j}=w_{i',j'}$. Dividing by $2$, we obtain
\[
(2^b-1)(i-i')+2^b(j-j')=0.
\]
Since $\gcd(2^b-1,2^b)=1$, there exists $k\in\mathbb Z$ such that
$i-i'=2^bk$ and 
$j-j'=-(2^b-1)k$.
If $(i,j)\neq(i',j')$, then $k\neq 0$, and hence
	$$
		m \ge \max\{i,i'\}\geq |i-i'|\geq 2^b.
	$$
Consequently, when $m<2^b$, all the exponents in the multinomial
expansion are distinct. Since there are $\binom{m+2}{2}$ admissible
pairs $(i,j)$, and only $(0,0)$ gives the zero weight, we obtain
	$$
		t=\tbinom{m+2}{2}-1=\tfrac{m(m+3)}{2}.
	$$
Finally, as $m=2^{\ell-b-1}$, the condition $m<2^b$ is equivalent
to $\ell\leq 2b$. Therefore, for $b+2\leq\ell\leq 2b$,
	$$
		t=2^{\ell-b-2}\bigl(2^{\ell-b-1}+3\bigr),
	$$
as claimed.
\end{proof}

We now give an example to illustrate the arithmetic collision boundary in the corollary.	
\begin{exam} \label{exam: collision}
Consider the SBR-code with $r=27$ and $\ell=5$. 
Here $r = 2^2 \cdot 7 - 1$, so $b=2$. Since $\ell = 5 > 4 = 2b$, the condition for weight uniqueness is violated. 
The reduction exponent is $m = 2^{5-2-1} = 4$. By Corollary \ref{coro: t-weights} $(b)$, the theoretical upper bound for the number of non-zero weights is $t \le \frac{4 \cdot 7}2 = 14$. 
	
Let us compute the actual weights. The base weights for the threshold trinomial (at $\ell_0 = 3$) are $u = 2^{2+1}-2 = 6$ and $v = 2^{2+1} = 8$. 
When expanding the trinomial 
$$ 1 + 2^b(r-1) x^{2^{b+1}-2} + (r^{2}-1-2^b(r-1)) x^{2^{b+1}} = 1 + 104x^6 + 624x^8$$ 
to the power of $m=4$, the resulting weights are linear combinations of the form 
	$$6i + 8j,$$ 
where $S = i+j$ is the number of non-zero terms drawn from the trinomial, with $1 \le S \le 4$. 
Grouping the generated weights by $S$, we obtain:
	\begin{itemize}
		\item $S = 1$: $\{6, 8\}$ (2 weights), \sk 
		
		\item $S = 2$: $\{12, 14, 16\}$ (3 weights), \sk
		
		\item $S = 3$: $\{18, 20, 22, \mathbf{24}\}$ (4 weights), \sk
		
		\item $S = 4$: $\{\mathbf{24}, 26, 28, 30, 32\}$ (5 weights).
	\end{itemize}

Notice the arithmetic collision at weight $24$, which is generated in two algebraically distinct ways: $(x^8)^3$ and $(x^6)^4$.
Because this weight overlaps, the actual number of distinct non-zero weights drops to $t = 13$. This perfectly validates the threshold condition in the proof: since the gap between the base weights is $8-6=2$, the minimal overlap requires exactly $4$ blocks of $6$ to equal $3$ blocks of $8$, which can only happen when $m \ge 2^b = 4$.
\end{exam}

\section{Explicit computations} \label{sec: examples}
Here we illustrate, in three examples, the results in the previous sections for the $r$-ary ICCs of length $2^\ell$ in the smallest cases, that is for $p=3,5,7$ and $\ell=2,3,4,5$. 
We will give the weight enumerators and the associated GP-graphs in tables. In particular, we can see the threshold codes and the exact number of non-zero weights $t$ for each code, illustrating the growth described in Corollary \ref{coro: t-weights}.

\begin{exam} \label{exam: SBR p=3}
Let $p=3$. We analyze the irreducible $r$-ary cyclic codes of length $n=2^\ell$ for $\ell \in \{2, 3, 4, 5\}$ across the extensions $r=3^2=9$ and $r=3^3=27$.
	
\noindent \textbf{Case $r=9$:} Here $r \equiv 1 \pmod 4$, taking the form $r = 2^b c + 1$ with $b=3$ and $c=1$.
For $\ell=2 \le b$, the code $\CC(2, 9)$ has parameters $[4, 1, 4]_9$.  It is a $1$-weight code ($t=1$). The GP-graph is $K_9$, with spectrum $\{[8]^1, [-1]^8\}$. By Theorem \ref{thm: WC(x) para SBR} $(a)$, $W_\CC(x) = 1 + 8x^4$.
For $\ell=3 \le b$, the code $\CC(1, 9)$ has parameters $[8, 1, 8]_9$ and $t=1$. The GP-graph remains $K_9$. The weight enumerator is 			
	$$ W_\CC(x) = 1 + 8x^8 .$$
For $\ell=4 > b$, the code $\CC(5, 9^2)$ has parameters $[16, 2, 8]_9$ and evolves into a $2$-weight code ($t=2^{4-3}=2$). The graph transitions to $H(2,9)$, with spectrum $\{[16]^1, [7]^{16}, [-2]^{64}\}$. By Theorem \ref{thm: WC(x) para SBR}~$(b)$, 
	$$ W_\CC(x) = (1 + 8x^8)^2 = 1 + 16x^8 + 64x^{16}. $$
	For $\ell=5 > b$, the code $\CC(205, 9^4)$ has parameters $[32, 4, 8]_9$ and becomes a $4$-weight code ($t=2^{5-3}=4$). The graph is $H(4,9)$, with spectrum explicitly given by $\{[32]^1, [23]^{32}, [14]^{384}, [5]^{2048}, [-4]^{4096}\}$. The weight enumerator expands to 
	$$ W_\CC(x) = (1 + 8x^8)^4 =1+32x^{8} +384x^{16}+2048x^{24}+4096x^{32}.$$

\noindent \textbf{Case $r=27$:} Here $r \equiv 3 \pmod 4$, taking the form $r = 2^b c - 1$ with $b=2$ and $c=7$.
For $\ell=2 \le b+1$, the code $\CC(182, 27^2)$ has parameters $[4, 2, 2]_{27}$. 
It is a $2$-weight code ($t=2$). The GP-graph is the semiprimitive graph $\G(14, 729)$, with spectrum $\{[52]^1, [25]^{52}, [-2]^{676}\}$. 
By $(d)$ in Theorem~\ref{thm: WC(x) para SBR}, 
$$ W_\CC(x) = 1 + 52x^2 + 676x^4. $$
For $\ell=3 \le b+1$, the code $\CC(91, 27^2)$ has parameters $[8, 2, 6]_{27}$ and $t=2$. The GP-graph is the semiprimitive graph $\G(7, 729)$ with spectrum $\{[104]^1, [23]^{104}, [-4]^{624}\}$. The weight enumerator is 
$$ W_\CC(x) = 1 + 104x^6 + 624x^8. $$
For $\ell=4 > b+1$, the code $\CC(33215, 27^4)$ has parameters $[16, 4, 6]_{27}$. Here $m = 2^{4-2-1}=2$, yielding a $5$-weight code ($t=5$). The graph is the Cartesian product $\square^2 \G(7, 729)$. By Theorem \ref{thm: WC(x) para SBR} $(f)$, 
$$ W_\CC(x) = (1 + 104x^6 + 624x^8)^2 = 1 + 208x^6 + 1248x^8 + 10816x^{12} + 129792x^{14} + 389376x^{16}. $$
For $\ell=5 > b+1$, the code $\CC(\frac{27^8-1}{32}, 27^8)$ has parameters $[32, 8, 6]_{27}$. Here $m = 2^{5-2-1}=4$. While the upper bound from Corollary \ref{coro: t-weights} gives $t \le \frac{4 \cdot 7}2 = 14$, this code length crosses the collision threshold $\ell > 2b$. The graph is $\square^4 \G(7, 729)$, yielding the weight enumerator:
\begin{align*} 
	W_\CC(x) &= (1 + 104x^6 + 624x^8)^4 \\
	&= 1 + 416x^6 + 2496x^8 + 64896x^{12} + 778752x^{14} + 2336256x^{16} + 4499456x^{18} \\
	&\quad + 80990208x^{20} + 485941248x^{22} + 1088868352x^{24} + 2807660544x^{26} \\
	&\quad + 25268944896x^{28} + 101075779584x^{30} + 151613669376x^{32}. 
\end{align*}
Notice that the weight $24$ is generated by both $(x^8)^3$ and $(x^6)^4$. Due to this arithmetic collision, the sequence merges, producing exactly a $13$-weight code (see Example \ref{exam: collision}).

Table~1 summarizes the properties for these $3$-ary codes.
\begin{table}[H] 
	\centering
	{\small \renewcommand{\arraystretch}{1.2}
		\begin{tabular}{|c|c|c|c|c|c|c|c|}
			\hline
			$r$ & $b$ & $\ell$ & Code $\CC$ & Parameters & $t$ & Graph $\G$ & $W_\CC(x)$ \\
			\hline
			\multirow{4}{*}{$9$} & \multirow{4}{*}{$3$}
			& $2$ & $\CC(2, 9)$ & $[4, 1, 4]$ & $1$ & $K_9$ & $1 + 8x^4$ \\
			&& $3$ & $\CC(1, 9)$ & $[8, 1, 8]$ & $1$ & $K_9$ & $1 + 8x^8$ \\
			&& $4$ & $\CC(5, 9^2)$ & $[16, 2, 8]$ & $2$ & $H(2, 9)$ & $(1 + 8x^8)^2$ \\
			&& $5$ & $\CC(205, 9^4)$ & $[32, 4, 8]$ & $4$ & $H(4, 9)$ & $(1 + 8x^8)^4$ \\
			\hline
			\hline
			\multirow{4}{*}{$27$} & \multirow{4}{*}{$2$}
			& $2$ & $\CC(182, 27^2)$ & $[4, 2, 2]$ & $2$ & $\G(14, 729)$ & $1 + 52x^2 + 676x^4$ \\
			&& $3$ & $\CC(91, 27^2)$ & $[8, 2, 6]$ & $2$ & $\G(7, 729)$ & $1 + 104x^6 + 624x^8$ \\
			&& $4$ & $\CC(33215, 27^4)$ & $[16, 4, 6]$ & $5$ & $\square^2 \G(7, 729)$ & $(1 + 104x^6 + 624x^8)^2$ \\
			&& $5$ & $\CC(\frac{27^8-1}{32}, 27^8)$ & $[32, 8, 6]$ & $13$ & $\square^4 \G(7, 729)$ & $(1 + 104x^6 + 624x^8)^4$ \\
			\hline
		\end{tabular}
		\caption{Parameters, graphs, and weight enumerators for $p=3$.}}
	\label{tab: SBR p3}
\end{table}
\end{exam}

\begin{exam} \label{exam: SBR p=5}
	Let $p=5$. We analyze the codes of length $n=2^\ell$ for $\ell \in \{2, 3, 4, 5\}$ across $r=5^2=25$ and $r=5^3=125$.
	
	\noindent \textbf{Case $r=25$:} Here $r \equiv 1 \pmod 4$, with $b=3$ and $c=3$.
	For $\ell=2 \le b$, $\CC(6, 25)$ is a $[4, 1, 4]_{25}$-code. The graph is $K_{25}$. By Theorem \ref{thm: WC(x) para SBR} $(a)$, 	
	$$ W_\CC(x) = 1 + 24x^4. $$
	For $\ell=3 \le b$, $\CC(3, 25)$ is a $[8, 1, 8]_{25}$-code. The graph is $K_{25}$, giving 
	$$ W_\CC(x) = 1 + 24x^8. $$
	For $\ell=4 > b$, $\CC(39, 25^2)$ is a $[16, 2, 8]_{25}$-code ($t=2$). The graph is $H(2, 25)$. By Theorem \ref{thm: WC(x) para SBR} $(b)$, 
	$$ W_\CC(x) = (1 + 24x^8)^2. $$
	For $\ell=5 > b$, $\CC(12207, 25^4)$ is a $[32, 4, 8]_{25}$-code ($t=4$). The graph is $H(4, 25)$, yielding 
	$$ W_\CC(x) = (1 + 24x^8)^4. $$ 
	The spectrum is $\{[96]^1, [71]^{96}, [46]^{3456}, [21]^{55296}, [-4]^{331776}\}$.
	
	\noindent \textbf{Case $r=125$:} Here $r \equiv 1 \pmod 4$, with $b=2$ and $c=31$.
	For $\ell=2 \le b$, $\CC(31, 125)$ is a $[4, 1, 4]_{125}$-code. The graph is $K_{125}$. The weight enumerator is 
	$$ W_\CC(x) = 1 + 124x^4. $$
	For $\ell=3 > b$, $\CC(1953, 125^2)$ is a $[8, 2, 4]_{125}$-code ($t=2$). The graph becomes $H(2, 125)$, with spectrum $\{[248]^1, [123]^{248}, [-2]^{15376}\}$. The enumerator is 
	$$ W_\CC(x) = (1 + 124x^4)^2=1 + 248x^4 + 15376x^{8}. $$
	For $\ell=4 > b$, $\CC(15258789, 125^4)$ is a $[16, 4, 4]_{125}$-code ($t=4$). The graph is $H(4, 125)$, with spectrum $\{[496]^1, [371]^{496}, [246]^{92256}, [121]^{7626496}, [-4]^{236421376}\}$. 
	The weight enumerator is 
	$$ W_\CC(x) = (1 + 124x^4)^4. $$
	For $\ell=5 > b$, $\CC(\frac{125^8-1}{32}, 125^8)$ is a $[32, 8, 4]_{125}$-code ($t=8$). 
	The graph is $H(8, 125)$, with spectrum $\{[992]^1, [867]^{992}, [742]^{430528}, \dots, [-8]^{124^8}\}$. 
	The enumerator is 
	$$ W_\CC(x) = (1 + 124x^4)^8. $$
	
	Table~2 summarizes the properties for these $5$-ary codes.
	\begin{table}[H] 
		\centering
		{\small	\renewcommand{\arraystretch}{1.2}
			\begin{tabular}{|c|c|c|c|c|c|c|c|}
				\hline
				$r$ & $b$ & $\ell$ & Code $\CC$ & Parameters & $t$ & Graph $\G$ & $W_\CC(x)$ \\
				\hline
				\multirow{4}{*}{$25$} & \multirow{4}{*}{$3$} 
				& $2$ & $\CC(6, 25)$ & $[4, 1, 4]$ & $1$ & $K_{25}$ & $1 + 24x^4$ \\
				&& $3$ & $\CC(3, 25)$ & $[8, 1, 8]$ & $1$ & $K_{25}$ & $1 + 24x^8$ \\
				&& $4$ & $\CC(39, 25^2)$ & $[16, 2, 8]$ &  $2$ & $H(2, 25)$ & $(1 + 24x^8)^2$ \\
				&& $5$ & $\CC(12207, 25^4)$ & $[32, 4, 8]$ & $4$ & $H(4, 25)$ & $(1 + 24x^8)^4$ \\
				\hline
				\hline
				\multirow{4}{*}{$125$} & \multirow{4}{*}{$2$}
				& $2$ & $\CC(31, 125)$ & $[4, 1, 4]$ & $1$ & $K_{125}$ & $1 + 124x^4$ \\
				&& $3$ & $\CC(1953, 125^2)$ & $[8, 2, 4]$ &  $2$ & $H(2, 125)$ & $(1 + 124x^4)^2$ \\
				&& $4$ & $\CC(15258789, 125^4)$ & $[16, 4, 4]$ &  $4$ & $H(4, 125)$ & $(1 + 124x^4)^4$ \\
				&& $5$ & $\CC(\frac{125^8-1}{32}, 125^8)$ & $[32, 8, 4]$ & $8$ & $H(8, 125)$ & $(1 + 124x^4)^8$ \\
				\hline
			\end{tabular}
			\caption{Parameters, graphs, and weight enumerators for $p=5$.}}
		\label{tab: SBR p5}
	\end{table}
\end{exam}

\begin{exam} \label{exam: SBR p=7}
	Let $p=7$. We explore the lengths $n=2^\ell$ for $\ell \in \{2, 3, 4, 5\}$ across $r=7^2=49$ and $r=7^3=343$.
	
	\noindent \textbf{Case $r=49$:} Here $r \equiv 1 \pmod 4$, with $b=4$ and $c=3$.
	For $\ell=2 \le b$, $\CC(12, 49)$ is a $[4, 1, 4]_{49}$-code. 
	The graph is $K_{49}$. By Theorem \ref{thm: WC(x) para SBR} $(a)$, $W_\CC(x) = 1 + 48x^4$.
	For $\ell=3 \le b$, $\CC(6, 49)$ is a $[8, 1, 8]_{49}$-code. 
	The graph is $K_{49}$, yielding 
	$$ W_\CC(x) = 1 + 48x^8. $$
	For $\ell=4 \le b$, $\CC(3, 49)$ is a $[16, 1, 16]_{49}$-code. 
	The graph is $K_{49}$, yielding 
	$$ W_\CC(x) = 1 + 48x^{16}. $$
	For $\ell=5 > b$, $\CC(75, 49^2)$ is a $[32, 2, 16]_{49}$-code  ($t=2$). 
	The graph becomes $H(2, 49)$. By Theorem~\ref{thm: WC(x) para SBR} $(b)$, 
	$$ W_\CC(x) = (1 + 48x^{16})^2. $$
	
	\noindent \textbf{Case $r=343$:} 
	Here $r \equiv 3 \pmod 4$, taking the form $r = 2^b c - 1$ with $b=3$ and $c=43$.
	For $\ell=2 \le b+1$, $\CC(29412, 343^2)$ is a $[4, 2, 2]_{343}$-code. The GP-graph is the semiprimitive graph $\G(172, 117649)$. By Theorem \ref{thm: WC(x) para SBR} $(d)$, 
	$$ W_\CC(x) = 1 + 684x^2 + 116964x^4. $$
	For $\ell=3 \le b+1$, $\CC(14706, 343^2)$ is a $[8, 2, 6]_{343}$-code. The GP-graph is the semiprimitive graph $\G(86, 117649)$. The weight enumerator is 
	$$ W_\CC(x) = 1 + 1368x^6 + 116280x^8. $$
	For $\ell=4 \le b+1$, $\CC(7353, 343^2)$ is a $[16, 2, 14]_{343}$-code. The GP-graph is the semiprimitive graph $\G(43, 117649)$. The weight enumerator is 
	$$ W_\CC(x) = 1 + 2736x^{14} + 114912x^{16}. $$
	For $\ell=5 > b+1$, $\CC(432540225, 343^4)$ is a $[32, 4, 14]_{343}$-code. Since $m = 2^{5-3-1}=2$, it is a $5$-weight code ($t=5$). The graph is the Cartesian product $\square^2 \G(43, 117649)$. By Theorem \ref{thm: WC(x) para SBR} $(f)$, 
	$$ W_\CC(x) = (1 + 2736x^{14} + 114912x^{16})^2. $$
	
	Table~3 summarizes the properties for these $7$-ary codes.
	\begin{table}[H] 
		\centering
		{\small	\renewcommand{\arraystretch}{1.2}
			\begin{tabular}{|c|c|c|c|c|c|c|c|}
				\hline
				$r$ & $b$ & $\ell$ & Code $\CC$ & Parameters &  $t$ & Graph $\G$ & $W_\CC(x)$ \\
				\hline
				\multirow{4}{*}{$49$} & \multirow{4}{*}{$4$}
				& $2$ & $\CC(12, 49)$ & $[4, 1, 4]$ & $1$ & $K_{49}$ & $1 + 48x^4$ \\
				&& $3$ & $\CC(6, 49)$ & $[8, 1, 8]$ &  $1$ & $K_{49}$ & $1 + 48x^8$ \\
				&& $4$ & $\CC(3, 49)$ & $[16, 1, 16]$ & $1$ & $K_{49}$ & $1 + 48x^{16}$ \\
				&& $5$ & $\CC(75, 49^2)$ & $[32, 2, 16]$ & $2$ & $H(2, 49)$ & $(1 + 48x^{16})^2$ \\
				\hline
				\hline
				\multirow{4}{*}{$343$} & \multirow{4}{*}{$3$}
				& $2$ & $\CC(29412, 343^2)$ & $[4, 2, 2]$ & $2$ & $\G(172, 343^2)$ & $1 + 684x^2 + 116964x^4$ \\
				&& $3$ & $\CC(14706, 343^2)$ & $[8, 2, 6]$ & $2$ & $\G(86, 343^2)$ & $1 + 1368x^6 + 116280x^8$ \\
				&& $4$ & $\CC(7353, 343^2)$ & $[16, 2, 14]$ & $2$ & $\G(43, 343^2)$ & $1 + 2736x^{14} + 114912x^{16}$ \\
				&& $5$ & $\CC(432540225, 343^4)$ & $[32, 4, 14]$ & $5$ & $\square^2 \G(43, 343^2)$ & $(1 + 2736x^{14} + 114912x^{16})^2$ \\
				\hline
			\end{tabular}
			\caption{Parameters, graphs, and weight enumerators for $p=7$.}}
		\label{tab: SBR p7}
	\end{table}
\end{exam}

\end{document}